\documentclass[letterpaper,10pt,conference]{ieeeconf}
\usepackage{multicol}
\usepackage[bookmarks=true]{hyperref}
\IEEEoverridecommandlockouts  
\usepackage{amsmath,amssymb,amsfonts}
\usepackage{bbm}
\usepackage{algorithmic}
\usepackage{graphicx}
\usepackage{textcomp}
\usepackage{setspace}
\usepackage{booktabs}
\usepackage{lipsum}

\usepackage{epsfig}
\usepackage{times} 
\usepackage[inkscapelatex=false]{svg}
\usepackage[linesnumbered,ruled,vlined]{algorithm2e} 

\usepackage[nolist, nohyperlinks]{acronym}
\acrodef{GNC}[GNC]{Graduated-Nonconvexity}
\acrodef{WLS}[WLS]{Weighted Least Squares}
\acrodef{TLS}[TLS]{Truncated Least Squares}
\acrodef{SVD}[SVD]{Singular Value Decomposition}
\acrodef{VO}[VO]{Visual Odometry}
\acrodef{VIO}[VIO]{Visual-Inertial Odometry}
\acrodef{SDF}[SDF]{Signed Distance Field}
\acrodef{SDFs}[SDFs]{Signed Distance Fields}
\acrodef{ESDF}[ESDF]{Euclidean Signed Distance Field}
\acrodef{CESDF}[C-ESDF]{Certified \ac{ESDF}}
\acrodef{TSDF}[TSDF]{Truncated Signed Distance Field}
\acrodef{RGBD}[RGBD]{RGB-Depth}
\acrodef{IMU}[IMU]{Inertial Measurement Unit}
\acrodef{EKF}[EKF]{Extended Kalman Filter}
\acrodef{VO}[VO]{Visual Odometry}
\acrodef{CVO}[C-VO]{Certified Visual Odometry}
\acrodef{ISS}[ISS]{Input-to-State}
\acrodef{SLAM}[SLAM]{Simultaneous Localization and Mapping}
\acrodef{SFC}[SFC]{Safe Flight Corridors}
\acrodef{FOV}[FoV]{field of view}
\acrodef{UAV}[UAV]{}
\acrodef{FOO}[FOO]{first-order optimality}
\acrodef{DOF}[DoF]{degree of freedom}
\acrodefplural{DOF}[DoFs]{degrees of freedom}
\acrodef{VTOL}[VTOL]{vertical takeoff and landing}

\usepackage{cite}

\usepackage{bm}
\usepackage{breqn}
\usepackage{amsmath}
\usepackage{amssymb}
\usepackage{tabularx}

\usepackage{amsthm}

\newcommand{\reals}{\mathbb{R}}
\newcommand{\R}{\reals}

\newcommand{\SO}{\mathbb{SO}}

\newcommand{\Lcal}{\mathcal{L}}

\newcommand{\Ucal}{\mathcal{U}}

\newcommand{\Wcal}{\mathcal{W}}
\newcommand{\Xcal}{\mathcal{X}}

\newcommand{\Zcal}{\mathcal{Z}}

\newcommand{\Lu}{\Lcal_u}

\newcommand{\Ll}{\Lcal_\lambda}
\newcommand{\Llx}{\Lcal_{\lambda_x}}
\newcommand{\Llz}{\Lcal_{\lambda_z}}

\newcommand{\Lul}{\begin{bmatrix}\Lu\\ \Ll\end{bmatrix}}

\newcommand{\od}{\operatorname{d}}

\newcommand{\eqn}[1]{\begin{align} #1 \end{align}}
\newcommand{\eqnN}[1]{\begin{align*} #1 \end{align*}}
\newcommand{\Aeqn}[1]{\begin{equation} \begin{aligned} #1 \end{aligned} \end{equation}}

\newcommand{\bmat}[1]{\begin{bmatrix}#1\end{bmatrix}}

\newcommand{\norm}[1]{\left\Vert #1 \right \Vert}

\theoremstyle{plain}
\newtheorem{theorem}{Theorem}
\newtheorem{corollary}{Corollary}

\newtheorem{problem}{Problem}
\newtheorem{definition}{Definition}

\theoremstyle{definition}
\newtheorem{assumption}{Assumption}

\DeclareMathOperator{\sgn}{sgn}

\theoremstyle{remark}

\newtheorem{example}{Example}

\makeatletter
\let\NAT@parse\undefined
\makeatother
\usepackage{hyperref}
\hypersetup{
colorlinks, 
    linkcolor={red!50!black},
    citecolor={blue!50!black},
    urlcolor={blue!80!black}
}
\usepackage{cleveref}

\crefformat{problem}{Problem~#2#1#3}
\crefformat{assumption}{Assumption~#2#1#3}

\title{\LARGE \bf Adaptive Control Allocation for Underactuated \\ Time-Scale Separated Non-Affine Systems }

\ifx\anonymized\undefined 
\author{Daniel M. Cherenson and Dimitra Panagou 
\thanks{This research was supported by the Center for Autonomous Air Mobility and Sensing (CAAMS), an NSF IUCRC, under Award Number 2137195, and an NSF CAREER under Award Number 1942907.}% <-this % stops a space
\thanks{All authors are with the Robotics Department, University of Michigan, Ann Arbor, MI, USA
        {\tt \{dmrc, dpanagou\}@umich.edu}}%
}
\else
\author{Author Names Omitted for Anonymous Review. Paper-ID}
\thanks{$^{*}$[redacted]}
\fi
\begin{document}

\maketitle

\begin{abstract}
Many robotic systems are underactuated, meaning not all degrees of freedom can be directly controlled due to lack of actuators, input constraints, or state-dependent actuation. This property, compounded by modeling uncertainties and disturbances, complicates the control design process for trajectory tracking. In this work, we propose an adaptive control architecture for uncertain, nonlinear, underactuated systems with input constraints. Leveraging time-scale separation, we construct a reduced-order model where fast dynamics provide virtual inputs to the slower subsystem and use dynamic control allocation to select the optimal control inputs given the non-affine dynamics. To handle uncertainty, we introduce a state predictor-based adaptive law, and through singular perturbation theory and Lyapunov analysis, we prove stability and bounded tracking of reference trajectories. The proposed method is validated on a \acf{VTOL} quadplane with nonlinear, state-dependent actuation, demonstrating its utility as a unified controller across various flight regimes, including cruise, landing transition, and hover.
\end{abstract}

\section{Introduction}
\label{sec:introduction}

Many robotic systems, including aerial robots~\cite{wu2023geometric}, underwater robots~\cite{lei2024disturbance}, and manipulators~\cite{pucci2015collocated} are underactuated. In this paper, we consider a broad class of uncertain underactuated systems where the mapping from control inputs to \acp{DOF} is not surjective, so not every \ac{DOF} can be directly commanded~\cite{underactuated}. Underactuation occurs when a system has fewer independent controls than \acp{DOF}, but is also caused by input constraints and state-dependent actuation. Because not every \ac{DOF} can be directly commanded, these systems cannot in general follow arbitrary trajectories which complicates control design in the presence of modeling uncertainty and environmental disturbances.

Some systems exhibit time-scale separation, meaning a subset of the system states evolve more quickly than the rest. Exploiting this property allows the use of reduced-order models, in which the fast subsystem reference is treated as a virtual input to the slow system, simplifying the overall control problem by augmenting the control input space. Singular perturbation theory provides formal treatment and analysis of time-scale separated control systems~\cite{kokotovic1999singular}.

Adaptive control provides a framework for handling uncertainties in dynamical systems by adjusting controller parameters online to maintain stability and tracking performance without requiring exact knowledge of the system~\cite{lavretsky2024robust}. While adaptive methods have been successfully applied to certain underactuated systems, much of the existing work has focused on control-affine manipulators with relatively simple actuation structures or without input constraints~\cite{pucci2015collocated,yang2021adaptive,sun2022singular}. Adaptive methods for non-affine systems with uncertainties in how the control inputs enter the dynamics have been considered~\cite{hovakimyan2008adaptive,yang2007adaptive}, but only for fully-actuated systems. Time-scale separation has been applied to adaptive control for non-minimum phase systems~\cite{eves2023aadaptive} and was extended to general singularly-perturbed systems in~\cite{eves2023badaptive}, but the results are derived for fully-actuated systems. The quadrotor is a well-studied underactuated platform with much adaptive control research~\cite{o2022neural,wu2025boldsymbol}, but these works take advantage of the simple structure of the quadrotor dynamics, limiting the application to more challenging non-affine systems with state-dependent actuation. The combination of underactuation, non-affine-in-control dynamics, and time-scale separation in adaptive control of a general class of systems remains largely unexplored.

% A further complication arises when a system possesses redundant actuators. In this case, there are infinitely many control inputs that can achieve a desired virtual control effect. The process of resolving this redundancy, known as control allocation, is typically cast as an optimization problem that selects actuator commands subject to constraints while minimizing some cost, often related to control energy or actuator usage.

% One structural property that can simplify the control of complex underactuated systems is time-scale separation. Many aerial and underwater vehicles, for example, exhibit fast and slow dynamics that can be separated through feedback control. By introducing virtual control inputs associated with the fast subsystem, the reduced slow dynamics can effectively be rendered overactuated. This enables the use of dynamic control allocation to consistently select actuator commands that achieve the desired closed-loop behavior.

In this paper, we propose an adaptive control architecture for uncertain, underactuated nonlinear systems that uses a time-scale separated control strategy. Our contributions are as follows: 
1) We design an adaptive control architecture for uncertain, non-affine in control, underactuated nonlinear systems with redundant, constrained control inputs using a time-scale separation property and dynamic control allocation, 2) We prove stability and bounded tracking error of a reference trajectory using singular perturbation theory and Lyapunov analysis, and 3) We demonstrate the algorithm on a \acf{VTOL} quadplane with highly nonlinear dynamics and state-dependent actuation, i.e., the system is underactuated in level flight and in hover and is overactuated during the landing transition phase.

% The paper is organized as follows. \Cref{sec:prelim} provides background on singular perturbation theory for multiple time-scale models. In \Cref{sec:prob}, we introduce the dynamics of the underactuated system and pose the trajectory tracking problem. \Cref{sec:method} details our proposed approach, and \Cref{sec:theory} backs up the method with theoretical analysis. Finally, \Cref{sec:sim} demonstrates the adaptive control architecture in a \ac{VTOL} quadplane landing scenario. 

\section{Preliminaries}
\label{sec:prelim}
% \subsection{Singular Perturbations for Multiple Time-Scale Models}

In this paper, we form a closed-loop three-time-scale model for the control of an uncertain underactuated system:
\begin{subequations}  
\vspace{-13pt}
\eqn{\label{eq:intro_slow}\dot x &= f(t,x,z,u,\epsilon), \quad x(0) = \xi(\epsilon) \\ \label{eq:intro_fast}\epsilon \dot z &= g(t,x,z,u,\epsilon), \quad z(0) = \eta(\epsilon) \\ \label{eq:intro_fastest}\epsilon^2 \dot u &= p(t,x,z,u,\epsilon), \quad u(0) = \upsilon(\epsilon)}
\end{subequations}
where $0 < \epsilon \ll 1$ is a parameter that represents the time-scale separation between the three subsystems. That is, $u$ evolves at a rate $\epsilon$ faster than $z$, which evolves at a rate $\epsilon$ faster than $x$. The initial conditions $\xi(\epsilon)$, $ \eta(\epsilon)$, and $\upsilon(\epsilon)$ are smooth functions of $\epsilon$. The fast time scale is $t' = \frac{t}{\epsilon}$ and the fastest time scale is $t'' = \frac{t'}{\epsilon} = \frac{t}{\epsilon^2}$. The system is in standard form if the algebraic equation $0 = p(t,x,z,u,0)$ has at least one isolated real root $u = q_i(t,x,z)$. We drop the subscript $i$ and only refer to a single root. Moreover, the algebraic equation $0 = g(t,x,z,q(t,x,z),0)$ has at least one isolated real root $z = h_j(t,x)$, and again we drop the subscript $j$.

The three-time-scale model is equivalent to two nested two-time-scale models. In the fastest time-scale $t''$, let $v = u - q(t,x,z)$. The boundary layer subsystem is given by $\frac{\od}{\od t''}v = p(t,x,z,v+q(t,x,z),0)$. Similarly, let $y = z - h(t,x)$. The boundary layer in the fast time-scale $t'$ is given by $\frac{\od}{\od t'}y = g(t,x,y + h(t,x),q(t,x,y+h(t,x)),0)$. The reduced system is the result of substituting the roots $h$ and $q$ into \eqref{eq:intro_slow}, yielding $\dot {\bar x} = f(t,x,h(t,x),q(t,x,h(t,x)),0)$. By applying Tikhonov's Theorem \cite[Theorem 11.2]{khalil2002nonlinear}, we can show that the reduced slow subsystem $\bar x(t)$ is a uniform approximation to the actual $\epsilon$-dependent solution to \eqref{eq:intro_slow}, $x(t,\epsilon)$. That is, $x(t,\epsilon) - \bar x(t) = O(\epsilon)$, where a vector function $f(t, \epsilon) \in \R^n$ is said to be $O(\epsilon)$ over the interval $[t_1,\infty)$ if there exist positive constants $k$ and $\epsilon^*$ such that $\norm{f(t, \epsilon)} \le k\epsilon \quad \forall \epsilon \in [0, \epsilon^*], ~ \forall t \in [t_1, \infty)$.

\section{Problem Formulation}
\label{sec:prob}
% Consider a MIMO nonlinear system in controllable canonical form \eqn{\label{eq:sys} &\dot x_i = A_{x_i} x_i + B_{x_i}\tau_{x_i}(x,z,u) \quad \forall i = 1, \ldots, n_x \\ &\dot z_j = A_{z_j}z_j + B_{z_j}\tau_{z_j}(x,z,u) \quad \forall j = 1, \ldots, n_z}
% where $x_i \in \Xcal_i \subset \R^{d_{x_i}}$, $z_j \in \Zcal_j \subset \R^{d_{z_j}}$
% \eqn{x = \bmat{x_1 \\ \vdots \\ x_{n_x}} \in \Xcal_1 \times \cdots \times \Xcal_{n_x} \subset \R^{N_x}, N_x = \sum_{i=1}^{n_x}d_{x_i}}
% \eqn{z = \bmat{z_1 \\ \vdots \\ z_{n_z}} \in \Zcal_1 \times \cdots \times \Zcal_{n_z} \subset \R^{N_z}, N_z = \sum_{j=1}^{n_z}d_{z_j}}
% $n_x + n_z$ is the total number of degrees of freedom of the system and $d_{x_i}$ is the relative degree of the degree of freedom $x_i$
Consider a MIMO nonlinear system in controllable canonical form with respect to a control input function $\tau(x,z,u)$. Note that $\tau$ is \emph{not} the actual control input, which we denote as $u$. The considered class of systems is of the form:
% \eqn{\label{eq:sys} &\dot x_i = A x_i + B\tau_{x_i}(x,z,u) \quad \forall i = 1,\ldots,n_x\\ &\dot z_j = Az_j + B\tau_{z_j}(x,z,u) \quad \forall j = 1,\ldots,n_z,}
\begin{subequations} \label{eq:sys}
\eqn{ \dot x = A_x x + B_x \tau_x(x,z,u) \\ \dot z = A_z z + B_z \tau_z(x,z,u)}
\end{subequations}
where $x \in \Xcal \subset \R^{dn_x}$ is the slow state vector, $z \in \Zcal \subset R^{dn_z}$ is the fast state vector, and $u \in \Ucal \subset \R^m$ is the control input. $d$ is the relative degree, which is uniform for $x$ and $z$. The control input function $\tau_x : \Xcal \times \Zcal \times \Ucal \to \R^{n_x}$ is a nonlinear function of the control input $u$ with $n_x$ components corresponding to $n_x$ \acfp{DOF} of the state $x$. The same applies to $\tau_z : \Xcal \times \Zcal \times \Ucal$ with $n_z$ \acp{DOF}. We define $A_x = A \otimes I_{n_x}$ and $A_z = A \otimes I_{n_z}$, where $\otimes$ is the Kronecker product of two matrices. Similarly, $B_x = B \otimes I_{n_x}$ and $B_z = B \otimes I_{n_z}$. $A$ and $B$ are defined as
 \eqn{A = \bmat{
0 & 1 & \cdots & 0 & 0 \\
\vdots & \vdots & \ddots & \vdots & \vdots \\
0 & 0 & \cdots & 0 & 1 \\
0 & 0 & \cdots & 0 & 0
} \in \R^{d\times d}, 
\quad
B = \bmat{
0 \\
\vdots \\
0 \\
1} \in \R^d.}
We assume that the full states $x$ and $z$ are measured. Input constraints are encoded by the polytope $\Ucal = \{u : Cu \le c\}, C \in \R^{n_c \times m}, c \in \R^{n_c}$ where $n_c>0$ is the number of individual input constraints. We further partition $z$ into $d$ components, denoted $z = \bmat{z_1^\top \cdots z_d^\top}^\top$, where each $z_i \in \R^{n_z}$, $i\in\{1,2,\dots,d\}$, and we will focus on $z_1$, which will become a virtual control input to the slow system.
% We write the combined state dynamics as \eqn{\dot x = A_x x + B_x \tau_x(x,z,u), \\ \dot z = A_z z + B_z \tau_z(x,z,u),} where $A_x = I_{n_x} \otimes A$, $B_x =I_{n_x} \otimes B$, $A_z = I_{n_z} \otimes A$, and $B_z = I_{n_z} \otimes B$. The symbol $\otimes$ denotes the Kronecker product of two matrices.
The combined state dynamics are $\dot\chi = A_\chi \chi + B_\chi \tau(x,z,u),$ where $A_\chi \in \R^{d(n_x + n_z)\times d(n_x+n_z)}$ is a block diagonal matrix of $A_x$ and $A_z$, and the same applies for $B_\chi \in \R^{d(n_x+n_z)\times n_x+n_z}$.

% where $x_i \in \R^d$ is a slow \ac{DOF}, $z_i \in \R^d$ is a fast \ac{DOF}, and $u \in \Ucal \subset \R^m$ is the vector of control inputs. $n_x + n_z$ is the total number of \acp{DOF}. The $l$-th component of the \ac{DOF} $x_i$ or $z_j$ is denoted as $[x_i]^l$ or $[z_j]^l$ where $l = 1,\ldots,d$. We assume the relative degree $d$ is the same for all \acp{DOF}.

The slow-fast state partition aids in the control design process, where singular perturbation theory confirms that treating the fast state as a control input to the slow state \acp{DOF} is valid. We further write the combined control input function as ${\tau(x,z,u) = [\tau_x(x,z,u)^\top,\tau_z(x,z,u)^\top]^\top.}$ We approximate the function $\tau : \Xcal \times \Zcal \times \Ucal \to \R^{n_x+n_z}$ as the sum of a known component $\tau_0$ and a linearly-parameterized uncertainty $\phi(x,z,u)W$, where $\phi : \Xcal \times \Zcal \times \Ucal \to \R^w$ is a vector of known locally Lipschitz nonlinear basis functions, and $W \in \Wcal \subset \R^w$ is a constant vector of \textit{unknown} parameters belonging to a known compact set $\Wcal$. The resulting form of $\tau$ is: \eqn{\label{eq:tau}\tau(x,z,u) := \tau_0(x,z,u) + \phi(x,z,u)W + \delta(x,z,u),}
where $\delta(x,z,u) \in R^{n_x+n_z}$ is a non-parametric uncertainty. Although $\delta$ is not estimated, we assume a known bound, denoted as $\delta_{\max} = \sup_{x,z,u}\norm{\delta(x,z,u)}$.

In our proposed approach, we design an estimator for $W$ and we use the estimated parameters $\widehat W$ in the controller to represent the current estimated $\tau$, denoted $\hat \tau$ for brevity: \eqn{\hat \tau = \hat\tau(x,z,u;\widehat W) := \tau_0(x,z,u) + \phi(x,z,u)\widehat W.} Note that we do not maintain an estimate of $\delta$ or use the bound $\delta_{\max}$ in our control design. We only use it for robustness analysis. The error between $\tau$ and $\hat \tau$ is \eqn{\tilde \tau := \tau - \hat \tau = \phi(x,z,u)\widetilde W + \delta(x,z,u),} where $\widetilde W = W - \widehat W$ is the parameter estimation error.

The focus of this paper is on underactuated systems. 
% We say the system \eqref{eq:sys} is underactuated if $\tau$ has a certain property, defined below:
\begin{definition}[Underactuated]\footnote{Trivially underactuated systems are those with fewer control inputs than degrees of freedom, i.e., $m < n_x+n_z$. However, even if $m = n_x + n_z$, some actuators may be redundant and there may be input constraints that limit the domain of the control input, $\Ucal$, causing underactuation. See \Cref{example} for a simple input-constrained underactuated system.}
The system \eqref{eq:sys} is underactuated if there exists a commanded output $\tau^c \in \R^{n_x+n_z}$ such that $\forall x \in \Xcal, z \in \Zcal, \widehat W \in \Wcal, u \in \Ucal$, $\hat\tau(x,z,u) \ne \tau^c$. That is, $\hat\tau(x,z,\cdot) : \Ucal \to \R^{n_x+n_z}$ is not surjective.
\end{definition}

% The reference model for the state $x$ is given as \eqn{\dot x^r = A^rx^r + B^rr} where $A_r$ is Hurwitz and $B_r$ is a column vector
% \eqn{A^r = \bmat{
% 0 & 1 & \cdots & 0 \\
% \vdots & \vdots & \ddots & \vdots \\
% 0 & 0 & \cdots & 1 \\
% k^r_0 & k^r_1 & \cdots & k^r_{d}
% }, B^r = \bmat{
% 0 \\
% \vdots \\
% 0 \\
% b^r}.
% } Since $A^r$ $k^r_i$
% Let $K_r = \bmat{k^r_0 & k^r_1 & \cdots & k^r_{d}}$ and assume that 

% \subsection{Trajectory Tracking}
Let $x^r(t)$ and $z^r(t)$ be time-varying, continuously-differentiable reference trajectories generated by a reference model or computed by a trajectory planner. In the following section, the proposed controller is designed so that the slow state $x$ tracks the reference trajectory $x^r(t)$ at all times, while the fast state $z$ tracks the reference trajectory $z^r(t)$ whenever possible. To track the desired trajectory, we design an adaptive control architecture to estimate the parameters $\widehat W$, which affects the selection of the control input $u$. Due to the uncertainty $\delta$ and the time-scale separation assumption, only bounded tracking error is guaranteed. We arrive at the overall problem statement:

\begin{problem}
\label{prob}
  Design a parameter adaptation law $\pi_W$ and a feedback controller $\pi_c$ that selects the control input $u$ for the underactuated system \eqref{eq:sys} such that the slow state tracking error $\norm{x(t)-x^r(t)}$ remains bounded $\forall t\geq 0$.
\end{problem}

% In the next section, we will propose an adaptive control architecture to solve \Cref{prob}.

\section{Method}
\label{sec:method}
\begin{figure}
  \label{fig:block_diag}
  \centering
  \includegraphics[width=\linewidth]{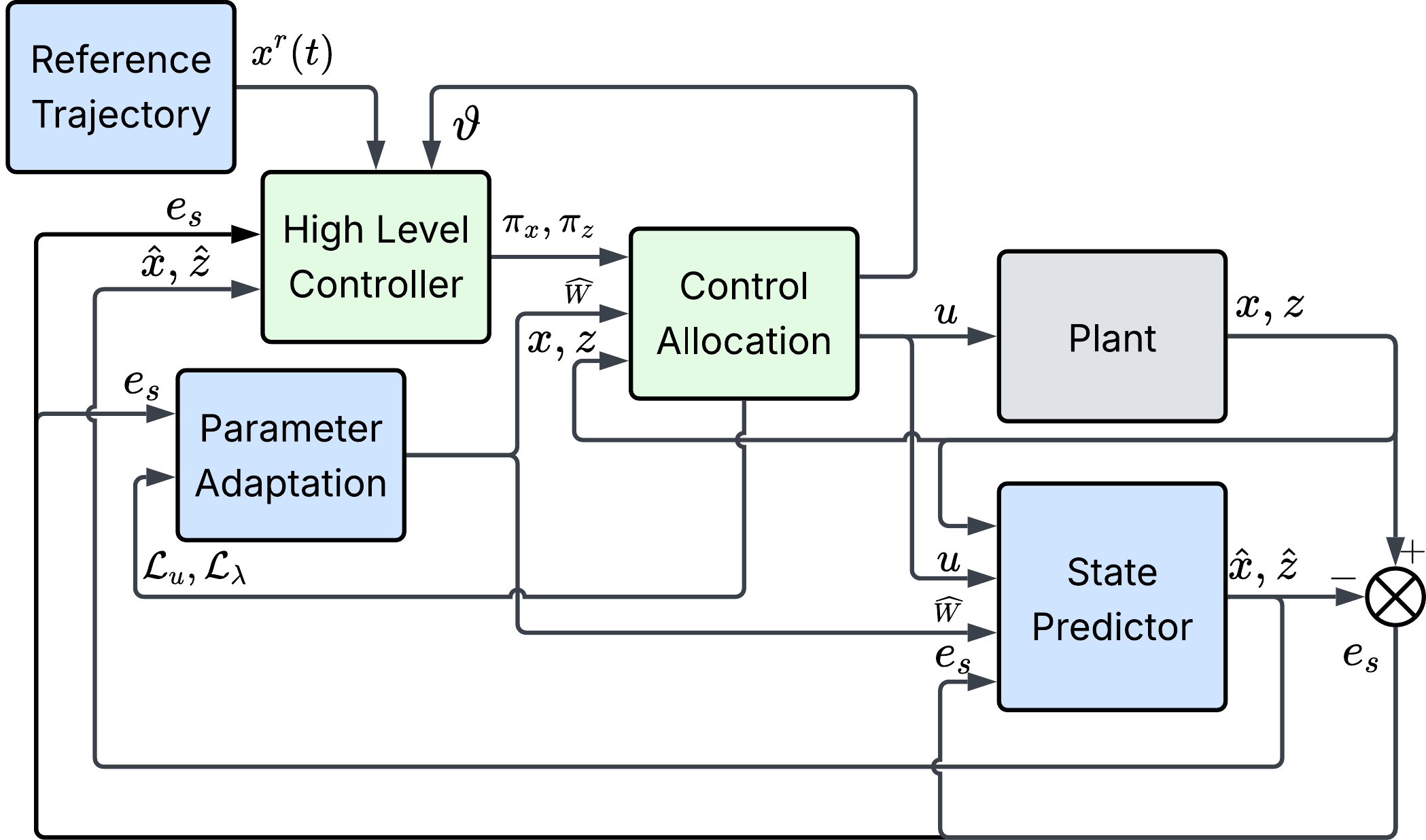}
  \caption{Block diagram of the proposed adaptive control architecture. The gray blocks are user-specified. The blue and green blocks are detailed in \Cref{sec:method}, and our specific contributions appear in the green blocks.}
  \vspace{-15pt}
\end{figure}

Tracking arbitrary trajectories for all \acp{DOF} is impossible because \eqref{eq:sys} is underactuated. Using a reduced-order model facilitates control design, since the fast state $z$ can be treated as a control input to the slow state $x$. %To achieve the high-level control command, we implement a control allocation algorithm acting as the inner loop controller. 
To establish stability and bounded tracking, we apply singular perturbation theory, which requires exponential stability of the tracking error system. This is difficult to show with an adaptive controller, so instead, we introduce a state predictor and modify the adaptation law for $\widehat W$ to reduce prediction error. A high-level controller reduces the error between the predictor and the reference trajectory, while the control allocation algorithm selects $u$ to realize the high-level command. Ultimately, the error between the slow state $x$ and the reference $x^r(t)$ is shown to be bounded.

To cast the closed-loop system in the form of \cref{eq:intro_slow,eq:intro_fast,eq:intro_fastest}, we employ a non-traditional adaptive architecture. Exponential stability of $x(t)-x^r(t)$ cannot be proven directly due to non-parametric uncertainty $\delta$ and lack of persistence of excitation. Instead, by controlling predictors $\hat x$ and $\hat z$ rather than the true states, we remove the uncertainty from the closed-loop dynamics. Inspired by~\cite{hovakimyan2008adaptive}, our approach has two steps: (i) adapt parameters to reduce $\hat x(t)-x(t)$, and (ii) apply high-level feedback to reduce $\hat x(t)-x^r(t)$.

Our contributions to the architecture in \Cref{fig:block_diag} are the high-level controller and control allocation blocks, along with theoretical guarantees of stability and bounded tracking.

\subsubsection{Reduced System Reformulation}

We seek to write the control function $\tau$ in a form that makes it surjective with respect to an augmented control input $\bar u$. Here, we define $\bar u$ as the original control input $u$ augmented by a command for the fast state $z$ such that the new form of $\tau$ is surjective. This step relies on the assumption that the fast state can be controlled much more quickly than the slow state. We apply singular perturbation theory to the final closed-loop system to verify that our approach is valid and gives the desired stability and boundedness properties.

The following example of a planar quadrotor shows that although the combined state model is underactuated due to input constraints, the partitioning of the slow and fast states results in a fully-actuated reduced system.

\begin{example}[Planar Quadrotor with Pusher Propeller]\label{example} ~ \\ Consider a simplified planar quadrotor with two vertical propellers $u_1,u_2$ and one lateral pusher propeller $u_3$. The slow state $x = \bmat{p_x & v_x & p_y & v_y}^\top$ and fast state $z = \bmat{\theta & \omega}^\top$. The simplified control function $\tau$ is \eqn{\tau(x,z,u) = \bmat{R(\theta)\bmat{u_3 \\-u_1-u_2} + \bmat{0 \\ 1} \\ u_1-u_2},} where $R(\theta) \in \SO(2)$. The map $\tau(x,z,\cdot)$ is not surjective because of the input constraints. In a hover where $x = 0, z = 0$, there does not exist a $u$ such that $\tau = \bmat{-0.1, 0, 0}^\top$ because $u_3$ cannot be negative, as the pusher propeller can only produce force in one direction. If, however, we assume that the fast subsystem settles to the command $z^c = \bmat{\theta^c & 0}^\top$ instantaneously, where $\theta^c$ is a virtual control input, then there exists an augmented control $\bar u = \bmat{u^\top & \theta^c}^\top$ that achieves the desired $\tau$. The virtual control input $\theta$ is the reference signal to the fast dynamics, which must be sufficiently fast for the reduced model assumption to be valid.
\end{example}

Formalizing this notion, let $\vartheta \in \Theta \subset \R^{n_z}$ be a vector of fast system commands. We treat the fast state $z$ as a virtual input to the slow state dynamics by assuming the fast state instantaneously achieves its command $z^c = \bmat{\vartheta^\top & 0 & \cdots & 0}^\top$. The reduced-order system model is $\dot{ x} = A_x x + B_x\hat\tau_x(x,z^c,u)$. We now define the augmented control input $\bar u = \bmat{u^\top \vartheta^\top}^\top \in \bar \Ucal := \Ucal \times \Theta = \{\bar u : \bar C \bar u \le \bar c\}$ where $\bar C$ and $\bar c$ encode the augmented polytopic input constraints on the virtual input. The reduced control function is \eqn{\bar\tau_x(x,\bar u) := \hat\tau_x(x,z^c,u).} It is required that the fast state $z$ be fully-actuated to track the command $\vartheta$. Formally, we assume the following:

\begin{assumption}
  $\bar\tau_x(x,\cdot) : \bar \Ucal \to \R^{n_x}$ and $\hat\tau_z(x,z,\cdot) : \Ucal \to \R^{n_z}$ are surjective $\forall x \in \Xcal, z \in \Zcal, \widehat W \in \Wcal$.
\end{assumption}
% \todo{write in singular perturbation form, it will become more clear}

We will use this reduced system for control design of the high-level controller of the slow system.
% As mentioned before, we will use a state predictor in the control architecture to achieve the desired stability and boundedness properties of the tracking error.
% Next, we introduce the state predictor, and then we discuss the high-level control algorithm.

\subsubsection{State Predictor}

We introduce state predictors $\hat x$ and $\hat z$ to place the relevant time-scale separated dynamics into a form to which singular perturbation theory applies. \Cref{fig:block_diag} depicts the state predictor below the plant dynamics because they have similar inputs $u$ and outputs $x/\hat x$ and $z/\hat z$. %If we used the true states in the singular perturbation analysis, we would be unable to show exponential stability of the reduced subsystem due to the presence of the non-parametric uncertainty $\delta$.
The slow state predictor $\hat x$ is a series-parallel model defined by the following update law: \eqn{\label{eq:xhat}\dot{\hat x} = A_xx + B_x\hat\tau_x(x,z,u) - A_{s_x} (x - \hat x),} where $A_{s_x}$ is a Hurwitz matrix in the form \eqn{A_{s,x} =  \bmat{
0 & I_{n_x} & \cdots & 0 \\
\vdots & \vdots & \ddots & \vdots \\
0 & 0 & \cdots & I_{n_x} \\ \multicolumn{4}{c}{[-K_{s_x}]}
},}
where $K_{s_x} = \bmat{k_{s_x,1} & k_{s_x,2} & \cdots & k_{s_x,d}} \otimes I_{n_x}$ with each element non-negative. The $\hat z$ update law is defined as: \eqn{\label{eq:zhat}\dot{\hat z} = A_z z + B_z\hat\tau_z(x,z,u) - A_{s_z}e_{s_z}.}

The combined predicted state is $\hat \chi = [\hat x^\top \hat z^\top]^\top$, and we define $e_s = \chi - \hat \chi$. The dynamics of $e_s$ are \eqn{\label{eq:es_dot}\dot e_{s} = A_{s} e_{s} + B\Phi(x,z,u)\widetilde W + \delta(x,z,u),} where $A_s = {\rm diag}(A_{s_x},A_{s_z})$. Next, we introduce the error between the reference trajectory $x^r(t)$ and the predicted $\hat x$. The dynamics of the reference trajectory are given by the linear system $\dot x^r = A_xx^r + B_x r(t)$, where $r(t) \in \R^{n_x}$ is the derivative of the $d$-th component of each \ac{DOF} of $x^r$.
% For example, if $x^r$ contains a position and velocity reference, then $r(t)$ is the corresponding acceleration.
The tracking error dynamics are \eqnN{\dot e_{x} &= A_x x + B_x\hat\tau_x(x,z,u) - A_{s_x} e_{s_x}- A_xx^r - B_x r(t) \\& = A_x e_x + B_x(\hat\tau_x(x,z,u) - r(t) - K_{s_x} e_{s_x}).} 

\subsubsection{High-Level Controller}
We now design a high-level controller for the slow states, $\pi_x(t,x,\hat x) : \R \times \Xcal \times \Xcal \to \R^{n_x}$, and for the fast states, $\pi_z(\vartheta,z,\hat z) : \Theta \times \Zcal \times \Zcal \to \R^{n_z}$. In \Cref{fig:block_diag}, the high-level controller lies between the reference trajectory and control allocation blocks.

The goal is to choose the control input $u$ and fast state command $\vartheta$ such that the dynamics of $e_x$ are exponentially stable in the form $\dot e_x = A_{r_x}e_x$, where $A_{r_x}$ is of the same form as $A_{s_x}$ and $K_{s_x}$ is replaced by $K_{r_x} = \bmat{k_{r_x,1} & k_{r_x,2} & \cdots & k_{r_x,d}} > 0$ element-wise. 

% \eqn{A_{r_x} = \bmat{
% 0 & 1 & \cdots & 0 \\
% \vdots & \vdots & \ddots & \vdots \\
% 0 & 0 & \cdots & 1 \\
% \multicolumn{4}{c}{[-K_{r_x}]}
% }\otimes I_{n_x}}
% The resulting slow state high-level control law $\pi_x$ is \eqn{\label{eq:slow_control} \pi_x(t,\hat x, x) = -K_{s_x}\underbrace{(x - \hat x)}_{e_{s_x}} - K_{r_x} \underbrace{(\hat x - x^r(t))}_{e_x} + r(t).}

The resulting slow state high-level control law $\pi_x$ is \eqn{\label{eq:slow_control} \pi_x(t,\hat x, x) = r(t)-K_{s_x}(x - \hat x) - K_{r_x} (\hat x - x^r(t)).}

The reduced-order model assumption requires that $z_1 = \vartheta$. The desired equilibrium of the fast state predictor must then be $\hat z_1 = \vartheta - e_{s_z,1}$. Hence, the desired dynamics of $\hat z$ are $\dot{\hat z} = A_{r_z}\hat z + B_zk_{r_z,1}(\vartheta - e_{s_z,1}),$ where $A_{r_z}$ is the fast state reference matrix, which has the same structure as $A_{r_x}$ with the last $n_z$ rows as $-K_{r_z}$, and $k_{r_z,1}$ is the first element of $K_{r_z}$. Here, we introduce the parameter $\epsilon$
% \eqn{\epsilon = (\frac{k_{r_x,1}}{k_{r_z,1}})^{\frac{1}{d}},}
as the time-scale separation between the slow and fast reference systems. In the next section, we use this specific form of $A_{r_z}$ to put the fast system into singularly-perturbed standard form.
% Let $\epsilon = (\frac{k_{r_x,1}}{k_{r_z,1}})^{\frac{1}{d}}$.
The fast reference matrix is defined as \eqn{A_{r_z} =\frac{1}{\epsilon}T(\epsilon)^{-1}\bar A_{r_x}T(\epsilon),} where $\bar A_{r_x} \in \R^{dn_z \times dn_z}$ has the same entries as $A_{r_x}$ and $T(\epsilon) = {\rm diag}(1,\epsilon, \ldots, \epsilon^{d-1}) \otimes I_{n_z}$, which results in $K_{r_z} = \bmat{\frac{1}{\epsilon^d}k_{r_x,1} & \frac{1}{\epsilon^{d-1}}k_{r_x,2} & \cdots & \frac{1}{\epsilon}k_{r_x,d}} \otimes I_{n_z}.$ 
% where $S(\epsilon) = {\rm diag}(1/\epsilon^d,1/\epsilon^{d-1}, \ldots, \epsilon)$. 
% \eqn{A_{r_z} = \bmat{
% 0 & 1 & \cdots & 0 \\
% \vdots & \vdots & \ddots & \vdots \\
% 0 & 0 & \cdots & 1 \\
% \multicolumn{4}{c}{[-K_{r_z}]}}}

% The resulting fast state feedback control constraint is then \eqn{\label{eq:fast_control} \pi_z(z,\hat z,\vartheta) = -K_{s_z}\underbrace{(z - \hat z)}_{e_{s_z}} - K_{r_z} \hat z + k_{r_z,1}(\vartheta -\underbrace{(z_1 - \hat z_1)}_{e_{s_z,1}}).}
The resulting fast state feedback control constraint is then \eqn{\label{eq:fast_control} \pi_z(z,\hat z,\vartheta) = k_{r_z,1}(\vartheta -(z_1 - \hat z_1))-K_{s_z}(z - \hat z) - K_{r_z} \hat z.}

With the high-level controllers $\pi_x$ and $\pi_z$ defined, the next step is to select the control input $u$ and virtual input $\vartheta$ that achieves $\hat\tau_x = \pi_x$ and $\hat\tau_z = \pi_z$.

\subsubsection{Control Allocation and Parameter Adaptation}

The reduced system may be overactuated due to the additional virtual control inputs, i.e., the $\bar u$ that achieves the desired $\hat\tau_x = \pi_x$ and $\hat\tau_z = \pi_z$ is not unique, which results in a control allocation problem. In \Cref{fig:block_diag}, the control allocation block acts as an inner loop controller, taking commands from the high-level controller and producing control inputs. The following time-varying control allocation problem returns the optimal $\bar u^*$ that minimizes the cost $J$ achieves the high-level control commands:
 \begin{subequations}
 % \vspace{-10pt}
 \label{eq:control_alloc}
\begin{align}
  \min_{\bar u\in\bar\Ucal} \quad & J(t,x, z,\bar u) \\ \text{s.t.} \quad & \pi_x(t,x,\hat x) = \bar\tau_x(x, \bar u) \\ &\pi_z(\vartheta, z, \hat z) = \hat\tau_z(x, z, u).
\end{align}
\end{subequations}
In addition to penalizing the real control inputs $u$, the deviation between the external reference $z_1^r(t)$ and the optimal command $\vartheta$ is penalized, and is augmented with log-barrier functions to encode the input constraints $\bar u \in \Ucal$ as $J(t,x,z,\bar u) = J_u(x,z,u) + \alpha(\vartheta - z_1^r(t))^2 - \beta\sum_i \log(-C_i \bar u + c_i)$, where $\alpha > 0$, $0<\beta \ll 1 $, $\bar C_i$ is the $i$-th row of $\bar C$ and $\bar c_i$ is the $i$-th element of $\bar c$.

The Lagrangian of \eqref{eq:control_alloc} is $\Lcal(t,x,\hat x,z,\hat z,\bar u, \lambda; \widehat W) = J(t,x,z,\bar u) + \lambda^\top(\pi(t,x,\hat x,z,\hat z,\vartheta)-\hat\tau(x,z,\bar u;\widehat W)),$ where $\lambda \in \R^{n_x+n_z}$ is a vector of Lagrange multipliers corresponding to the $n_x + n_z$ constraints and $\pi = \bmat{\pi_x^\top & \pi_z^\top}^\top$. The \ac{FOO} vector is \eqn{\Lul = \bmat{\nabla_{u}\Lcal(t,x,\hat x,z,\hat z,\bar u,\lambda;\widehat W) \\ \nabla_\lambda \Lcal(t,x,\hat x,z,\hat z,\bar u;\widehat W)}.} Note that $\Ll$ is simply the constraint violation. Higher order derivatives also have a subscript for the relevant differential, e.g., $\nabla_u\Lu = \Lu{}_u$. The bordered Hessian is $\mathbb{H}(t,x,\hat x,z,\hat z,\bar u,\lambda;\widehat W) = \bmat{\Lu{}_u & \Ll{}_u \\ \Lu{}_\lambda & 0}.$ Unless otherwise specified, we abuse notation and drop the arguments when writing derivatives of $\Lcal$ and $\mathbb{H}$. The following assumptions ensure the optimization problem is well-conditioned.
\begin{assumption}
\label{assump:hessian}
  The objective function $J : \R \times \R^{dn_x} \times \R^{dn_z} \times \R^{m+n_z}$ is twice differentiable and radially unbounded in $\norm{\bar u}$. The Lagrangian $\Lcal$ is twice differentiable and there exist constants $k_2 > k_1 > 0$ such that $\forall t, x,\hat x, z, \hat z, \bar u, \lambda, W, \quad k_1I < \Lu{}_{u} < k_2I$.
\end{assumption}

\begin{assumption}
\label{assump:surjective}
  There exist constants $k_4>k_3>0$ such that $\forall x \in \Xcal, z \in \Zcal, \bar u \in \bar \Ucal, \widehat W \in \Wcal$, $k_3I < \nabla_{\bar u}\hat\tau(x,z,\bar u) \nabla_{\bar u}\hat\tau(x,z,\bar u)^\top < k_4I$.
\end{assumption}

% \begin{assumption}
%     \label{assump:feasible}
%     The optimal control allocation problem \eqref{eq:control_alloc} is always feasible with respect to the constraint $\bar u^* \in \bar \Ucal$.
% \end{assumption}

%\begin{remark}
If Assumptions \ref{assump:hessian} and \ref{assump:surjective} hold, then $\Lu = 0$ and $\Ll = 0$ is a sufficient condition for $(\bar u, \lambda)$ to be a strict local minimum~\cite[Theorem 12.6]{wright1999numerical}.
%\end{remark}

Instead of fully solving the optimization problem at discrete intervals, we construct a dynamical system whose solution converges to the optimal solution of control input $\bar u$ and Lagrange multipliers $\lambda$, based on~\cite{tjonnaas2008adaptive}. $\bar u$ and $\lambda$ are updated such that the \ac{FOO} vectors $\Lu$ and $\Ll$ asymptotically converge to zero, solving \eqref{eq:control_alloc}. %This approach is advantageous because it integrates well with an adaptation law for $\widehat W$.
% We introduce the following dynamical system that will asymptotically solve the optimal control allocation problem~\eqref{eq:control_alloc} while adapting the estimated parameters
The proposed adaptive control allocation update equations are
\begin{equation}
\begin{aligned}
\label{eq:ul_update}
\bmat{\dot {\bar u} \\ \dot \lambda} = -\mathbb{H}^{-1}(\Gamma_{u\lambda}\Lul + u_{\rm ff})
\end{aligned}
\end{equation}
\begin{equation}
  \begin{aligned}
  \label{eq:w_update}
    \dot{\widehat W} = {\rm Proj}\bigg(\widehat W, \Gamma_W\phi^\top B_\chi^\top\bigg(\Gamma_e e_s + \bmat{\Lcal_{\chi u} \\ \Lcal_{\chi\lambda}}^\top\Lul\bigg)\bigg),
    \end{aligned}
\end{equation}
where $u_{\rm ff}$ is a feed-forward term that compensates for the time evolution of the optimal solution, defined as $u_{\rm ff} =\bmat{\Lcal_{t u} \\ \Lcal_{t\lambda}} + \bmat{\Lcal_{\chi u} \\ \Lcal_{\chi\lambda}}(A_\chi\chi + B_\chi\hat\tau) + \bmat{\Lcal_{\hat \chi u} \\ \Lcal_{\hat \chi\lambda}}\dot{\hat \chi} + \bmat{\Lcal_{W u} \Lcal_{W\lambda}}\dot{\widehat W}.$ Recall that $\phi(x,z,u)$ is the known regressor for the uncertainty $W$ and is a function of the states and control input $u$. The projection operator ${\rm Proj}(\cdot,\cdot) : \R^w \times \R^w \to \R^w$ is defined in~\cite[Chapter 11]{lavretsky2024robust}\footnote{When applied to the parameter update law, the projection operator guarantees that the estimated parameters are bounded for all time.%, where the convex set is defined by sub-level sets of convex function $f_W : \R^w \to \R$~\cite[Chapter 11]{lavretsky2024robust}
}.
% ${\rm Proj}(\cdot,\cdot) : \R^w \times \R^w \to \R^w$ is \Aeqn{%&{\rm Proj}(W, \Gamma y) \\ &:= 
% \Gamma\begin{cases} y - \frac{\nabla f_W (\nabla f_W)^\top}{\norm{\nabla f_W}^2_\Gamma}\Gamma y f_W, ~ {\rm if} ~ f_W > 0 \land y^\top\Gamma\nabla f_W>0 \\ y, \quad {\rm otherwise}\end{cases}}
The matrices $\Gamma_{u\lambda} = {\rm diag}(\Gamma_u, \Gamma_\lambda)$, $\Gamma_W$, and $\Gamma_e$ are symmetric positive-definite gains. Assumption \ref{assump:surjective} ensures that $\bar u$ is always feasible with respect to the input constraints $\bar \Ucal$, which prevents \eqref{eq:ul_update} from ``blowing up'' due to the log barrier functions in $J$.
% \subsection{Projection Operator for Adaptive Control}

The feed-forward term $u_{\rm ff}$ accounts for the time evolution of the optimal solution $(\bar u^*,\lambda^*)$ because the control allocation problem~\eqref{eq:control_alloc} is parametric with respect to all variables. \Cref{th:adapt} in the next section details the derivation of the update laws. The estimated parameters are an auxiliary input to the control allocation and state predictor in \Cref{fig:block_diag}, whose outputs are fed back in the parameter update law \eqref{eq:w_update}.
%\subsection{Summary}

In summary, the proposed approach is designed to solve \Cref{prob}, comprising the state predictors \eqref{eq:xhat} and \eqref{eq:zhat}, the high-level control laws \eqref{eq:slow_control} and \eqref{eq:fast_control}, the optimal control allocation update laws for $\bar u$ and $\lambda$ in \eqref{eq:ul_update}, and the parameter adaptation law \eqref{eq:w_update}. Next, we show that the closed-loop system indeed solves \Cref{prob}.

\section{Theoretical Arguments}
\label{sec:theory}
% \subsubsection{Singular Perturbations}
We reformulate the three subsystems into standard form with three levels of time-scale separation. The signals from the adaptation law, $e_{s_x}(t)$, $e_{s_x}(t)$, and $\widehat W(t)$, along with the reference trajectory $x^r(t)$, are considered external time-varying signals in the singularly-perturbed closed-loop system.
% The goal is to show that under sufficient time-scale separation, the quasi-static relaxation of the fast states in the slow controller is valid. 
We summarize the coupled dynamics of $(e_x, \hat z,\bar u, \lambda)$:
\Aeqn{\label{eq:reformed_e_x}\dot e_x &= A_{r_x} e_x + B_x(\tau_x(e_x + x^r(t) + e_{s_x}(t),\hat z + e_{s_z}(t),u) \\ &- \bar\tau_x(e_x + x^r(t) + e_{s_x}(t),\vartheta,u) - \Llx(t,e_x,\hat z,\bar u))}\eqn{\dot {\hat z} &= A_{r_z} \hat z + B_z(k_{r_z,1} \vartheta - \Llz(t,e_x,\hat z,u)) \\ \label{eq:ul_sp} \bmat{\dot {\bar u} \\ \dot \lambda} &= -\mathbb{H}^{-1}\Gamma_{u\lambda}\Lul - u_{\rm ff}} where the dynamics of $e_x$ and $z$ are rewritten to include the application of the feedback laws \eqref{eq:slow_control} and \eqref{eq:fast_control}, respectively. The $\Llx$ and $\Llz$ terms are the constraint violations from the control allocation problem \eqref{eq:control_alloc}. The $\tau_x -\bar \tau_x$ term in \eqref{eq:reformed_e_x} is the mismatch between the full- and reduced-order models. The following analysis shows that both the $\Ll$ and $\tau_x$ mismatch terms vanish when $\epsilon=0$.

We also apply a state transformation to the fast subsystem to put it into standard form. Define $\zeta := T(\epsilon)\hat z$, where we recall that $T(\epsilon) = {\rm diag}(1,\epsilon, \ldots, \epsilon^{d-1}) \otimes I_{n_z}$. Then we have \eqn{\dot \zeta & = T(\epsilon)\dot{\hat z} = T(\epsilon)A_{r_z}T(\epsilon)^{-1}\zeta + T(\epsilon)B_z(k_{r_z,1}\vartheta - \Llz) \nonumber\\ &= \frac{1}{\epsilon}\bar A_{r_x}\zeta+\frac{1}{\epsilon}B_z(k_{r_x,1}\vartheta - \epsilon^{d-1}\Llz),} where $\bar A_{r_x}$ has the same entries as $A_{r_x}$ and the dimension of $A_{r_z}$, which puts the $e_x$ and $\zeta$ systems on the same time-scale. Multiplying by $\epsilon$ gives: \eqn{\epsilon\dot \zeta = \bar A_{r_x}\zeta + B_z(k_{r_x,1}\vartheta - \epsilon^{d}\Llz),} which transforms the fast subsystem into standard form and brings it onto the same time-scale as the slow subsystem.

We make the following assumption when analyzing the singularly-perturbed system.
\begin{assumption} \label{assump:sens} The sensitivity of the optimal solution with respect to the transformed state $\zeta$ satisfies $\frac{\partial}{\partial \zeta_1}\vartheta^*(t,e_x,\zeta)\bigg|_{\zeta=\zeta^*} < 1, ~ \forall t, e_x \in \Xcal$.
\end{assumption}

Let $\Gamma_{u\lambda}$ be selected such that $\underline\mu(\Gamma_{u\lambda}) = \frac{1}{\epsilon^2}\underline\mu(A_{r_x})$. Let $\bar\Gamma_{u\lambda} = \epsilon^2\Gamma_{u\lambda}$ and rewrite \eqref{eq:ul_sp} as %\eqn{\bmat{\dot{\bar u} \\ \dot \lambda} = -\mathbb{H}^{-1}(\frac{1}{\epsilon^2}\bar\Gamma_{u\lambda} \Lul + u_{\rm ff}).}
% Multiplying both sides by $\epsilon^2$ gives 
\eqn{\epsilon^2\bmat{\dot{\bar u} \\ \dot \lambda} = -\mathbb{H}^{-1}(\bar\Gamma_{u\lambda} \Lul + \epsilon^2 u_{\rm ff}).} The three-time-scale separated system is now in standard form. With the $\epsilon$-scaling, all systems evolve at the speed of the slow $x$ subsystem. The rescaled subsystems are: \Aeqn{\label{eq:slow_sp}&\dot e_x = f(t,e_x,\zeta,\bar u, \epsilon) := A_{r_x} e_x \\ &+ B_x(\tau_x(e_x+x^r(t)+e_{s_x}(t),T(\epsilon)^{-1}\zeta + e_{s_z}(t),u) \\ &- \bar\tau_x(e_x+x^r(t)+e_{s_x}(t),\vartheta,u) - \Llx) %, \quad e_x(0) = \xi(\epsilon)
}
\Aeqn{ \label{eq:fast_sp}&\epsilon\dot \zeta = g(t,e_x,\zeta,\bar u, \epsilon) := \bar A_{r_x}\zeta + B_z(k_{r_x,1}\vartheta - \epsilon^{d}\Llz)
% , \\& \zeta(0) = \eta(\epsilon)
} \Aeqn{\label{eq:fastest_sp} & \epsilon^2\bmat{\dot{\bar u} \\ \dot \lambda} = p(t,e_x,\zeta,\bar u, \lambda,\epsilon) := -\mathbb{H}^{-1}(\bar\Gamma_{u\lambda} \Lul + \epsilon^2u_{\rm ff}) %, \\& \bmat{\bar u(0) \\ \lambda(0)} = \upsilon(\epsilon)
.}
%where the initial conditions $\xi(\epsilon)$, $\eta(\epsilon)$, and $\upsilon(\epsilon)$ are smooth functions of $\epsilon$.
   
In \eqref{eq:fastest_sp}, setting $\epsilon = 0$ gives $-\mathbb{H}^{-1}\bar \Gamma_{u\lambda} \Lul = 0$. Assumptions \ref{assump:hessian} and \ref{assump:surjective} imply that $\mathbb{H}$ is always invertible, meaning $\mathbb{H}^{-1}\bar \Gamma_{u\lambda}$ can be divided out. The resulting equation is $\Lul = 0$, whose root is clearly a locally optimal solution to \eqref{eq:control_alloc}. Finally, let $q(t,e_x,\zeta) = \bmat{\bar u^* \\ \lambda^*}$ be a root of \eqref{eq:fastest_sp}, which is a strict local minimum of \eqref{eq:control_alloc} and hence an isolated root. Shifting the origin by $q$ results in the following boundary layer system: \eqn{\label{eq:fastest_bl}\frac{\od}{\od t''}\bmat{\bar u - \bar u^* \\ \lambda - \lambda^*} = -\mathbb{H}^{-1}\bar\Gamma_{u\lambda}\Lul}

Replacing $\bar u$ with $\bar u^*$ and $\lambda$ with $\lambda^*$ sets $\Ll = 0$: \Aeqn{\dot e_x &= A_{r_x} e_x+B_x(\tau_x(e_x + x^r(t)+e_{s_x}(t),\zeta,u^*) \\ &- \bar\tau_x(e_x + x^r(t)+e_{s_x}(t),\vartheta^*,u^*))}
\eqn{\label{eq:reduced_fast}&0 = \bar A_{r_x}\zeta + B_zk_{r_x,1}\vartheta^*(t,e_x,\zeta)} 
The isolated root of \eqref{eq:reduced_fast} is $\zeta^* = \bmat{h(t,e_x)^\top & 0 & \cdots & 0}^\top$, where 
% $h(t,e_x)$ is the root of the equation $\zeta_1 = \vartheta^*(t,e_x,\zeta)$ such that
$h(t,e_x) = \vartheta^*(t,e_x,h(t,e_x))$. Let $y = \zeta - \zeta^*$. The boundary layer subsystem for $\eqref{eq:reduced_fast}$ is \Aeqn{\label{eq:fast_bl}\frac{\od}{\od t'}y &= \bar A_{r_x}(y + \zeta^*)\\ &+ B_zk_{r_x,1}\vartheta^*(t,e_x + x^r(t)+e_{s_x}(t),y + \zeta^*).}
Finally, we replace $z_1$ with $\zeta_1=\vartheta^*$ in the resulting reduced slow subsystem, which cancels the $\tau_x$ terms, leaving \eqn{\label{eq:reduced_slow}\dot e_x = A_{r_x}e_x.}

\begin{theorem}
  \label{th:tikh}
% Let \Cref{assump:tikhonov1} hold.
Let $e_x(t,\epsilon)$ be the solution to \eqref{eq:slow_sp} and $\bar e_x(t)$ be the solution to \eqref{eq:reduced_slow} on the interval $[t_0,\infty)$. Then there exists a positive constant $\epsilon^*$ such that for all $0 < \epsilon < \epsilon^*$, $e_x(t,\epsilon) - \bar e_x(t) = O(\epsilon)$ holds uniformly for $t \in [t_0,\infty)$.
\end{theorem}

\begin{proof}
We apply Tikhonov's Theorem \cite[Theorem 11.2]{khalil2002nonlinear}, which requires a regularity condition and exponential stability of the boundary layer and reduced systems. We assume that the regularity condition holds.

First, we show that the boundary layer systems are exponentially stable. The linearization of \eqref{eq:fastest_bl} evaluated at $\bar u = \bar u^*$ and $\lambda = \lambda^*$ is $\nabla_{\bar u, \lambda}(-\mathbb{H}^{-1}\bar\Gamma_{u\lambda}\Lul)$. Simplifying, we have \eqnN{& -\nabla_{\bar u,\lambda}\mathbb{H}^{-1} \bar\Gamma_{u\lambda}\underbrace{\bmat{\Lu(\bar u^*,\lambda^*) \\ \Ll(\bar u^*,\lambda^*)}}_0 - \mathbb{H}^{-1} \bar\Gamma_{u\lambda}\underbrace{\nabla_{\bar u, \lambda}\Lul}_{\mathbb{H}} ,}
% \\ &= -\mathbb{H}^{-1} \bar\Gamma_{u\lambda} \mathbb{H}, 
resulting in the linear system $\frac{\od}{\od t''}\bmat{\bar u - \bar u^* \\ \lambda - \lambda^*} = -\mathbb{H}^{-1} \bar\Gamma_{u\lambda} \mathbb{H}\bmat{\bar u - \bar u^* \\ \lambda - \lambda^*}$. The origin is locally exponentially stable because $-\mathbb{H}^{-1} \bar\Gamma_{u\lambda} \mathbb{H}$ is negative definite.

To show exponential stability of the origin of \eqref{eq:fast_bl}, we linearize the system about $y = 0$. Note that $\nabla_{y} = \nabla_\zeta$. The linearization is $F := \nabla_{y}(\bar A_{r_x}(y + \zeta^*) + B_zk_{r_x,1}\vartheta^*(t,e_x,y + \zeta^*) = \bar A_{r_x} + B_zk_{r_x,1}\nabla_{\zeta}\vartheta^*(t,e_x,y+\zeta^*)$. Since $B_z$ is non-zero in the last element, it suffices to check the the last row of $F$, which is $-K_{r_x} + k_{r_x,1}\nabla_{\zeta}\vartheta^*(t,e_x,\zeta)|_{\zeta=\zeta^*}$. To ensure $F$ is Hurwitz, each element of the last row of $F$ must be negative, which follows from \Cref{assump:sens}. Hence, the origin of the boundary layer system \eqref{eq:fast_bl} is locally exponentially stable. The reduced subsystem \eqref{eq:reduced_slow} is clearly exponentially stable. Applying Tikhonov's Theorem, we have $e_x(t,\epsilon) - \bar e_x(t) = \hat x(t,\epsilon) - x^r(t) - \exp(A_{r_x}t)e_x(0) = O(\epsilon)$. 
\end{proof}

\Cref{th:tikh} essentially states that using the proposed high-level control and control allocation update laws, the tracking error $\hat x(t) - x^r(t)$ is made arbitrarily small by reducing $\epsilon$.

% \subsubsection{Adaptation and Control Allocation Stability and Boundedness}
%Now that we have shown that sufficient time-scale separation between the slow and fast high-level controllers results in bounded error between $\hat x(t)$ and $x^r(t)$, 
We next show that the error $e_s(t)$ is bounded.

\begin{theorem}
\label{th:adapt}
The closed-loop system of \cref{eq:sys,eq:xhat,eq:zhat,eq:ul_update,eq:w_update}
has ultimately bounded state prediction error $e_s(t)$ and first-order optimality vectors $\Lu(t)$ and $\Ll(t)$, with a bound proportional to $\delta_{max}$. Moreover, the parameter estimation error $\widetilde W(t)$ is ultimately bounded.
\end{theorem}
\begin{proof}
Define the candidate Lyapunov-like function 
\eqnN{
  & V = \frac{1}{2}\bigg(e_s^\top \Gamma_e e_s + \Lul^\top\Lul + \widetilde W^\top\Gamma_W^{-1}\widetilde W\bigg).
}
The time derivative of the \ac{FOO} vector is \eqnN{\bmat{\dot{\Lu} \\ \dot{\Ll}} % &= \bmat{\Lcal_{tu} + \Lcal_{\chi u} \dot \chi + \Lcal_{\hat\chi u} \dot {\hat\chi} + \Lu{}_u \dot {\bar u} + \Ll{}_u \dot \lambda + \Lcal_{W u}\dot{\widehat W} \\ \Lcal_{t\lambda} + \Lcal_{\chi\lambda}\dot \chi + \Lcal_{\hat\chi\lambda} \dot {\hat\chi} + \Lu{}_\lambda \dot {\bar u} + \Lcal_{W\lambda}\dot{\widehat W}} \\
 &= \underbrace{\bmat{\Lu{}_u & \Ll{}_u \\ \Lu{}_\lambda & 0}}_{\mathbb H}\bmat{\dot u \\ \dot\lambda} + \bmat{\Lcal_{tu} \\ \Lcal_{t\lambda}} + \bmat{\Lcal_{\chi u} \\ \Lcal_{\chi\lambda}}\dot \chi \\&+ \bmat{\Lcal_{\hat\chi u} \\ \Lcal_{\hat\chi\lambda}}\dot{\hat \chi} + \bmat{\Lcal_{W u} \\ \Lcal_{W\lambda}}\dot{\widehat W},}
% The goal is to show that $\dot{V}$ is negative outside a compact set, which will define the ultimate bounds on the relevant signals. 
and time derivative of $V$ is
\eqnN{\dot{V}&=%e_s^\top \Gamma_e \dot e_s + \Lul^\top\bmat{\dot{\Lu} \\ \dot{\Ll}} + \widetilde W \Gamma_W^{-1} \dot {\widetilde W} \\ &=
e_s^\top \Gamma_e(A_se_s + B_\chi(\phi\widetilde W + \delta)) \\ & + \Lul^\top\bigg( {\mathbb H}\bmat{\dot {\bar u} \\ \dot\lambda} + \bmat{\Lcal_{tu} \\ \Lcal_{t\lambda}} + \bmat{\Lcal_{\hat\chi u} \\ \Lcal_{\hat\chi\lambda}}\dot {\hat \chi} + \bmat{\Lcal_{\hat\vartheta u} \\ \Lcal_{\hat\vartheta\lambda}}\dot{\widehat W} \\ &+ \bmat{\Lcal_{\chi u} \\ \Lcal_{\chi\lambda}} \underbrace{A_\chi\chi + B_\chi\hat\tau}_{\rm known} + \underbrace{B_\chi(\phi\widetilde W + \delta)}_{\rm unknown}\bigg) - \widetilde W^\top\Gamma_W^{-1} \dot{\widehat W},}
where we used $\dot{\widetilde W} = \dot W - \dot{\widehat W}=-\dot{\widehat W}$ and plugged in \eqref{eq:es_dot} and split $\dot \chi$ into known and unknown parts. % as $\dot \chi = \underbrace{A_\chi\chi + B_\chi\hat\tau}_{\rm known} + \underbrace{B_\chi(\phi\widetilde W + \delta)}_{\rm unknown}$.
% where $A_\chi$ and $B_\chi$ are block-diagonal combinations of $A_x$, $A_z$, $B_x$, and $B_z$. 
Define $Q_e = -\frac{1}{2}(\Gamma_eA_s + A_s^\top\Gamma_e)$, meaning $Q_e \succ 0$. Then grouping terms
% , \eqnN{\dot V & = -e_s^\top Q_e e_s + \Lul^\top\bigg( {\mathbb H}\bmat{\dot {\bar u} \\ \dot\lambda} + \bmat{\Lcal_{tu} \\ \Lcal_{t\lambda}} \\& + \bmat{\Lcal_{\chi u} \\ \Lcal_{\chi\lambda}}(A_\chi\chi+ B_\chi\hat \tau) + \bmat{\Lcal_{\hat\chi u} \\ \Lcal_{\hat\chi\lambda}}\dot {\hat\chi} + \bmat{\Lcal_{W u} \\ \Lcal_{W\lambda}}\dot{\widehat W}\bigg) \\ & + \widetilde W^\top\bigg(-\Gamma_W^{-1} \dot{\widehat W} + \phi^\top B_\chi^\top(\Gamma_e e_s + \bmat{\Lcal_{\chi u} \\ \Lcal_{\chi\lambda}}^\top\Lul)\bigg) \\ & + (\Gamma_e e_s + \bmat{\Lcal_{\chi u} \\ \Lcal_{\chi\lambda}}^\top\Lul)^\top B_\chi\delta.}
and substituting the update laws for $\dot {\bar u}$, $\dot \lambda$, and $\dot {\widehat W}$ from \cref{eq:ul_update,eq:w_update}, we have
\eqnN{\label{eq:vdot}
    \dot V \le -e_s^\top Q_e e_s - \Lu^\top\Gamma_u\Lu - \Ll^\top\Gamma_\lambda\Ll \\+ (e_s^\top\Gamma_e + \Lu^\top\Lcal_{\chi u} + \Ll^\top\Lcal_{\chi\lambda}) B_\chi\delta.}
The projection operator ensures that the term $-\widetilde W^\top\Gamma_W^{-1} \dot{\widehat W}$ is always negative semi-definite~\cite[Chapter 11.4]{lavretsky2024robust}. Let $\underline\mu(\cdot)$ be the smallest eigenvalue of a matrix. Applying Cauchy-Schwartz and Young's inequalities, we have \eqnN{e_s^\top \Gamma_e B_\chi \delta \le \frac{1}{2}\underline\mu(Q_e)\norm{e_s}^2 + \frac{\norm{\Gamma_e B_\chi}^2\norm{\delta}^2}{2\underline\mu(Q_e)}, \\
\Lu^\top \Lcal_{\chi u} B_\chi \delta \le \frac{1}{2}\underline\mu(\Gamma_u)\norm{\Lu}^2 + \frac{\norm{\Lcal_{\chi u} B_\chi}^2\norm{\delta}^2}{2\underline\mu(\Gamma_u)}, \\
\Ll^\top \Lcal_{\chi\lambda} B_\chi \delta \le \frac{1}{2}\underline\mu(\Gamma_\lambda)\norm{\Ll}^2 + \frac{\norm{\Lcal_{\chi\lambda} B_\chi}^2\norm{\delta}^2}{2\underline\mu(\Gamma_\lambda)}.
} 

Plugging in $\delta_{\max}^2$ as the upper bound of $\norm{\delta}^2$, we conclude that
% \eqnN{\dot V \le -\underline\mu( Q_e) \norm{e_s}^2 -\underline\mu( \Gamma_u) \norm{\Lu}^2 -\underline\mu(\Gamma_\lambda) \norm{\Ll}^2 \\ + \frac{1}{2}\underline\mu(Q_e)\norm{e_s}^2 + \frac{\norm{\Gamma_e B_\chi}^2\delta_{\max}^2}{2\underline\mu(Q_e)} + \frac{1}{2}\underline\mu(\Gamma_u)\norm{\Lu}^2 \\+ \frac{\norm{\Lcal_{\chi u} B_\chi}^2\delta_{\max}^2}{2\underline\mu(\Gamma_u)} + \frac{1}{2}\underline\mu(\Gamma_\lambda)\norm{\Ll}^2 + \frac{\norm{\Lcal_{\chi\lambda} B_\chi}^2\delta_{\max}^2}{2\underline\mu(\Gamma_\lambda)} \\ = \frac{1}{2}(-\underline\mu( Q_e) \norm{e_s}^2 -\underline\mu( \Gamma_u) \norm{\Lu}^2 -\underline\mu(\Gamma_\lambda) \norm{\Ll}^2 \\ + D\delta_{\max}^2),} 
$\dot V$ is non-positive when $(e_s,\Lu,\Ll)$ is outside the set \eqnN{\Omega = \bigg\{(e_s,\Lu,\Ll)
% \in \R^{d+d} \times \R^{m+n_z} \times \R^{n_x+n_z} \times \R^w
: \norm{e_s} \le \sqrt{\frac{D}{\underline \mu(Q_e)}}\delta_{\max} \\ \land \norm{\Lu} \le \sqrt{\frac{D}{\underline \mu(\Gamma_u)}}\delta_{\max}\land \norm{\Ll} \le \sqrt{\frac{D}{\underline \mu(\Gamma_\lambda)}}\delta_{\max} \bigg\}} where $D := (\frac{\norm{\Gamma_e B_\chi}^2}{\underline\mu(Q_e)} + \frac{d_{\chi u}^2}{\underline\mu(\Gamma_u)} + \frac{d_{\chi \lambda}^2}{\underline\mu(\Gamma_\lambda)})$, $d_{\chi u} := \sup_t\norm{\Lcal_{\chi u} B_\chi}$, and $d_{\chi \lambda} := \sup_t\norm{\Lcal_{\chi \lambda} B_\chi}$. Hence, $(e_s,\Lu,\Ll)$ enters the set $\Omega$ in finite time, i.e., ultimately bounded. The bounds on $e_s(t)$, $\Lu(t)$, and $\Ll(t)$ are proportional to $\delta_{\max}$. Furthermore, the parameter estimation error $\widetilde{W}$ is bounded due to the projection operator.
\end{proof}
The following corollary solves Problem \ref{prob}.
\begin{corollary}
  $x(t)$ asymptotically tracks $x^r(t)$ with bounded error on the order of $\epsilon$ and $\delta_{\max}$.
\end{corollary}
\begin{proof}
  From Theorem \ref{th:adapt}, we showed that the signal $e_{s_x}(t)$ is ultimately bounded by $\norm{e_s} \le \sqrt{\frac{D}{\underline \mu(Q_e)}}\delta_{\max}$, where $D$ is defined in Theorem \ref{th:adapt}. Since $e_{s_x}(t) = x(t) - \hat x(t)$ and using the results of Theorem \ref{th:tikh}, we have $x(t) - x^r(t) = e_s(t) + \exp(A_{r_x}t)e_x(0) + O(\epsilon)$, asymptotically, $\norm{x(t) - x^r(t)} \le \sqrt{\frac{D}{\underline \mu(Q_e)}}\delta_{\max} + k\epsilon$, where $k$ is defined in \Cref{sec:prelim}.
\end{proof}

\section{Simulation}
\label{sec:sim}
\begin{figure*}[t]
  \centering
  \begin{minipage}{0.48\linewidth}
    \centering
    \includegraphics[width=\linewidth]{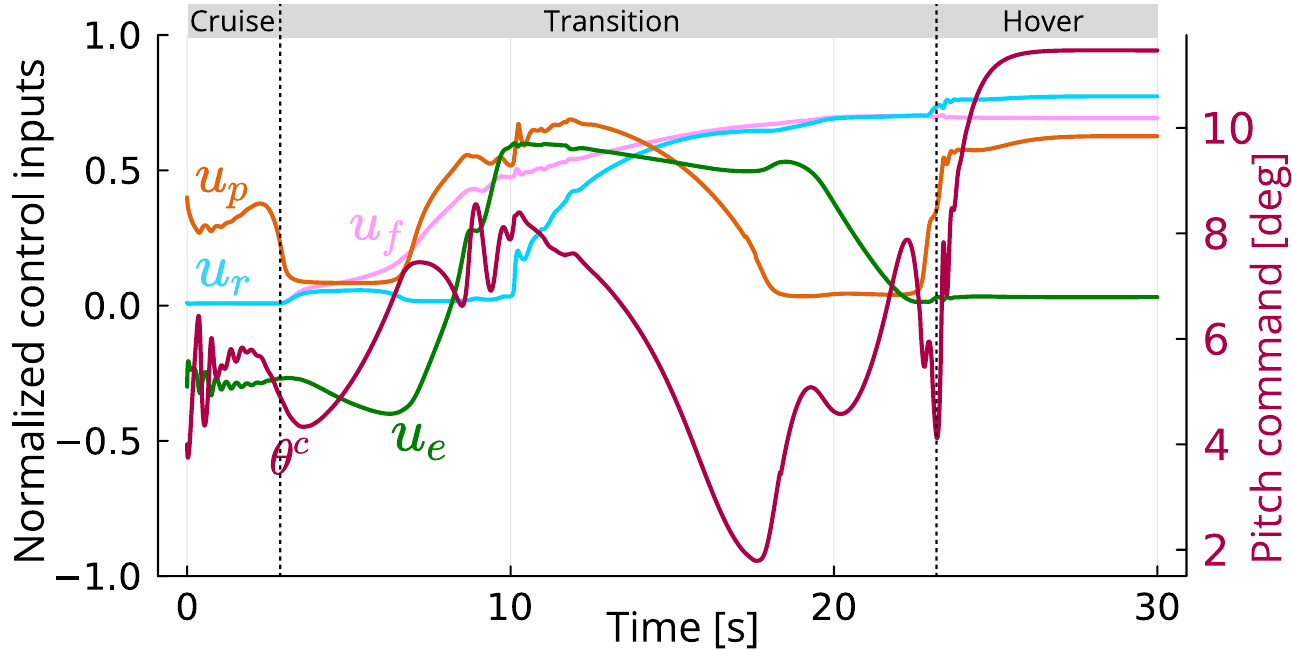}
    \caption{Control inputs over time during the three \ac{VTOL} landing phases: cruise, transition, and hover. }
    \label{fig:input}
  \end{minipage}\hfill
  \begin{minipage}{0.48\linewidth}
    \centering
    \includegraphics[width=\linewidth]{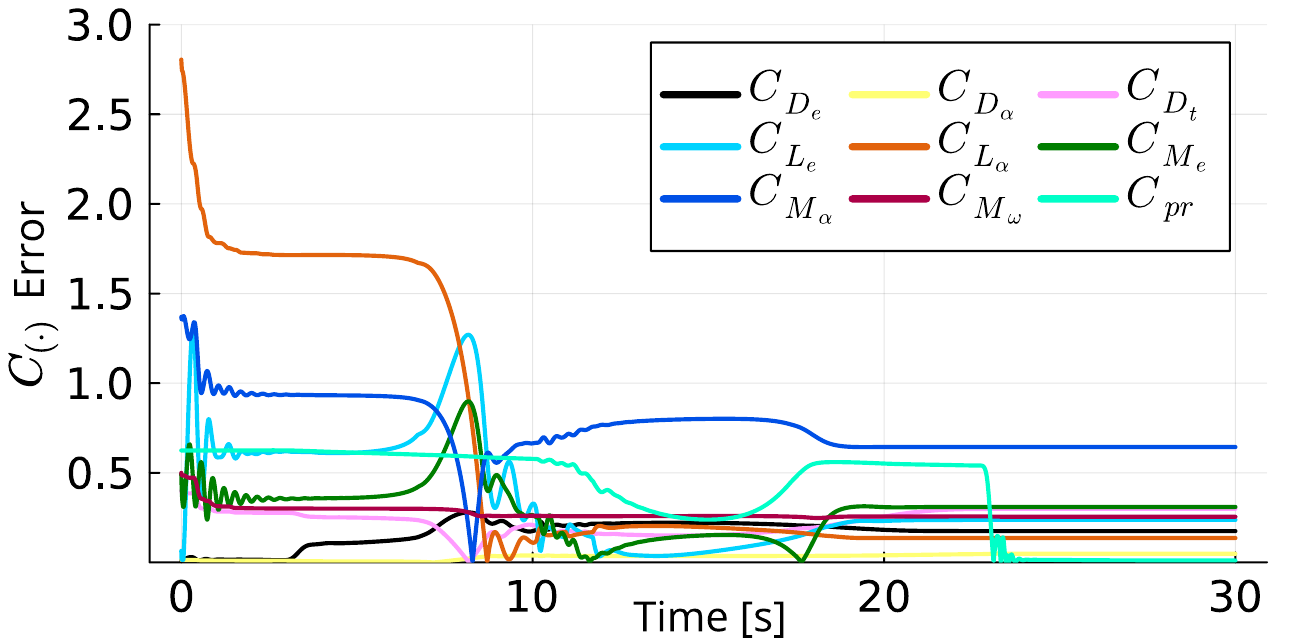}
    \caption{Parameter estimation error over time. Initial estimates are initialized at 1.5 times the true values.}
    \label{fig:params}
  \end{minipage}
  \vspace{-15pt}
\end{figure*}

Consider a 3-DOF planar \ac{VTOL} quadplane\footnote{The notation in this example is unrelated to previous sections.}:
\eqn{ &\dot p = v, \quad m\dot v = mg + R(\theta)F_b(v, \theta, u), \nonumber \\ &\dot \theta = \omega, \quad J\dot \omega = M(v,\theta, \omega, u), \nonumber } where $p = [p_x,p_z]^\top \in \R^2$ is the position in the inertial frame ($p_z$ points downward), $v = [v_x,v_z]^\top \in \R^2$ is the velocity in the inertial frame, $\theta \in [-\pi,\pi]$ is the pitch angle in radians, $\omega \in \R$ is the pitch angular velocity, and $u = [u_f, u_r, u_p, u_e]^\top \in \Ucal = [0,1]\times[0,1]\times[0,1]\times[-1,1]$ is the control input which contains the front vertical propeller speed $u_f$, rear vertical propeller speed $u_r$, pusher propeller speed $u_p$, and elevator deflection $u_e$. All control inputs are normalized by their respective limits. $m >0$ is the vehicle mass, $J > 0$ is the moment of inertia, and $g = [0,g_z]^\top$ is the gravitational acceleration. $R(\theta) \in \SO(2)$ is a rotation matrix between the body frame and the inertial frame. $F_b$ is the sum of propulsive and aerodynamic forces in the body frame. $M$ is the sum of propulsive and aerodynamic moments. $F_b$ and $M$ are nonlinear functions of the states and control inputs.

Here, the slow system is the translational variables and the fast system is the angular variables. Treating the pitch angle as a virtual input, the \ac{VTOL} becomes overactuated in all phases of flight. In a hover while maintaining a positive pitch angle, the \ac{VTOL} is fully-actuated without the pitch angle virtual command, which means the pitch angle can be chosen arbitrarily, as shown in the landing scenario below.

The aerodynamics are nonlinear functions of the propeller speeds, including interactions between propellers, rotor drag, and stall dynamics. Aerodynamic coefficients are denoted by $C_{(\cdot)} \in \R$. The airspeed is $V = \norm{v}$ and the angle of attack is $\alpha = \theta - \tan^{-1}(v_2/v_1)$. $\rho > 0$ is the air density and $S > 0$ is the wing area. 
$F_b(v,\theta,u) = R(\alpha)0.5\rho V^2\bmat{-C_D \\ C_L } + \bmat{T_{\max,p}u_p^2 \\ T_{\max,v}(u_f^2 + C_{pr}(V,u_p)u_r^2)}$. $T_{\max,(\cdot)} > 0$ is the thrust at maximum propeller speed. 
$C_{pr}(V,u_p)$ is a degradation factor of the rear propeller, as a function of the pusher propeller speed and airspeed, which introduces non-linearities in the control input, as $C_{pr}(V,u_p) = 1 - (1 - C_{pr}^{\rm nom})(1-V/V_{\max})u_p^2$. The lift and drag coefficients are

$C_D = C_{D_0} + C_{D_\alpha}\alpha^2 + C_{D_e}u_e + C_{D_t}(u_f+u_r)$ and $C_L = + C_{L_e}u_e + C_{L_\omega}\omega + (1-\sigma(\alpha))(C_{L_0} + C_{L_\alpha}\alpha) + \sigma(\alpha)2\sgn(\alpha)\sin^2(\alpha)\cos(\alpha),$ where $\sigma(\alpha)$ is a sigmoid function that blends the linear aerodynamics at low angle of attack with the nonlinear stall aerodynamics at high angle of attack~\cite{beard2012small}. The pitch moment dynamics are $M(v,\theta, \omega, u) = 0.5\rho V^2S\bar c C_M + l_v T_{\max, v}(u_f^2 - C_{pr}(V,u_p)u_r^2)$ where $C_M = C_{M_e}u_e + C_{M_\omega}\omega + (1-\sigma(\alpha))(C_{M_0} + C_{M_\alpha}\alpha) - \sigma(\alpha)2\sgn(\alpha)\sin^2(\alpha)\cos(\alpha)$. All coefficients and simulation parameters are available in our public code repository\footnote{\href{https://github.com/dcherenson/adaptive-control-underactuated.git}{https://github.com/dcherenson/adaptive-control-underactuated.git}}.
% \begin{table}[h]
%   \centering
%   \caption{Aerodynamic Coefficients and Simulation Parameters}

% \begin{tabular}{|c|c|c|c|c|c|c|c|}
% \hline
%   $C_{D_e}$& $C_{D_\alpha}$ & $C_{D_t}$ & $C_{L_e}$ & $C_{L_\alpha}$ & $C_{M_e}$ & $C_{M_\alpha}$ & $C_{M_\omega}$ \\
%   \hline
% \end{tabular}
%   \label{tab:coeffs}
% \end{table}

The landing reference trajectory contains three segments: cruise at constant altitude and speed, transition at constant deceleration and descent rate, and hover at a constant descent rate. The position reference trajectory (in meters) is \eqn{p_x^r(t) = \begin{cases}
  260 + 20t, \quad &0 \le t < 3 \\
  260 + 20t - 0.5(t-t_1)^2, \quad &3 \le t < 23 \\
  0, \quad &t_2 \le t
\end{cases},} and $p_z^r(t) = 100$ if $0\le t < 3$ or $p_z^r(t) = 100 - 0.75t$ if $t \ge 3$. A $\tanh$ smoothing function is applied to switch between components of the trajectory. This reference trajectory is designed to be dynamically feasible with respect to the control input limits. 
% In practice, slack variables for each constraint can be added to the control allocation problem to prevent infeasibility, acting as phantom actuators.

The pitch angle reference (in degrees) is $\theta^r(t) = 0$ if $0 \le t < 23$ or $\theta^r(t) = 15$ if $t \ge 23$. The objective function $J(t,\bar u) = \frac{1}{2}u^\top Q u + 25(\theta^c - \theta^r(t))^2 - 0.01\sum_{i=1}^{5}\log((-\bar u_i^{\min} + \bar u_i)(\bar u_i^{\max} - \bar u_i)),$ where $Q = \textrm{diag}(10,10,1,0.1)$ and penalizes the vertical propeller input much more than the pusher propeller and the elevator.
% Constraints on $\theta^c$ are $\pm 20$ deg and the constraints on $u$ are defined by $\Ucal$.

The reference gain for the position system is $K_{r,p} = \bmat{0.5 & 0.707}$ and introducing the time-scale separation $\epsilon$, the reference gain for the pitch system is $K_{r,\theta} = \bmat{0.5\epsilon^{-2} & 0.707\epsilon^{-1}}$. The state predictor gain is $K_s = \bmat{5 & 5}$ for both the slow and fast systems.

The adapted parameters are nine aerodynamic parameters: $C_{D_e}$, $C_{D_\alpha}$, $C_{D_t}$, $C_{L_e}$, $C_{L_\alpha}$, $C_{M_e}$, $C_{M_\alpha}$, $C_{M_\omega}$, and a modification of $C_{pr}$ using a simplified model, which introduces approximation error. Additionally, we omit the $C_{L_\omega}$ component of lift in the approximated dynamics to further add approximation error. The adaptation gain $\Gamma_W$ is chosen as $0.1 I_9$ and the state prediction error adaptation gain $\Gamma_e$ is $100 I_6$. The optimization gain $\Gamma_{u\lambda}$ is set to $50 I_8$.
% Note that here we are not using $\epsilon$ in the selection of $\Gamma_{u\lambda}$. Fixing a large gain means the control inputs $\bar u$ will converge to the optimal value more quickly.

The first and second order derivatives of the Lagrangian are computed with automatic differentiation in Julia. The update laws for $\widehat W$, $\bar u$, and $\lambda$ are run at 100 Hz using a first-order Euler discretization for solving the ODEs.

All parameters are initialized at 50\% above their true values and $\epsilon = 0.2$. \Cref{fig:input,fig:params} show the control inputs and parameters adapting during the automatic transition between forward flight and hover. The vertical propeller inputs are initially zero during cruise flight because the \ac{VTOL} wings generate sufficient lift to track the vertical reference. As the \ac{VTOL} slows, the vertical propeller inputs increase to their steady-state hover value as the aerodynamic lift fades due to low airspeed. During the transition, the elevator loses effectiveness as well. The pitch command does not follow the reference pitch command until the hover is reached, showing the switch from being underactuated during the transition to being fully-actuated in the hover. The parameter error is shown to be bounded and decreasing during the landing, but parameter error is not guaranteed to go to zero because persistency of excitation is not assumed.

\begin{figure}[t]
  \centering
  \includegraphics[width=\linewidth]{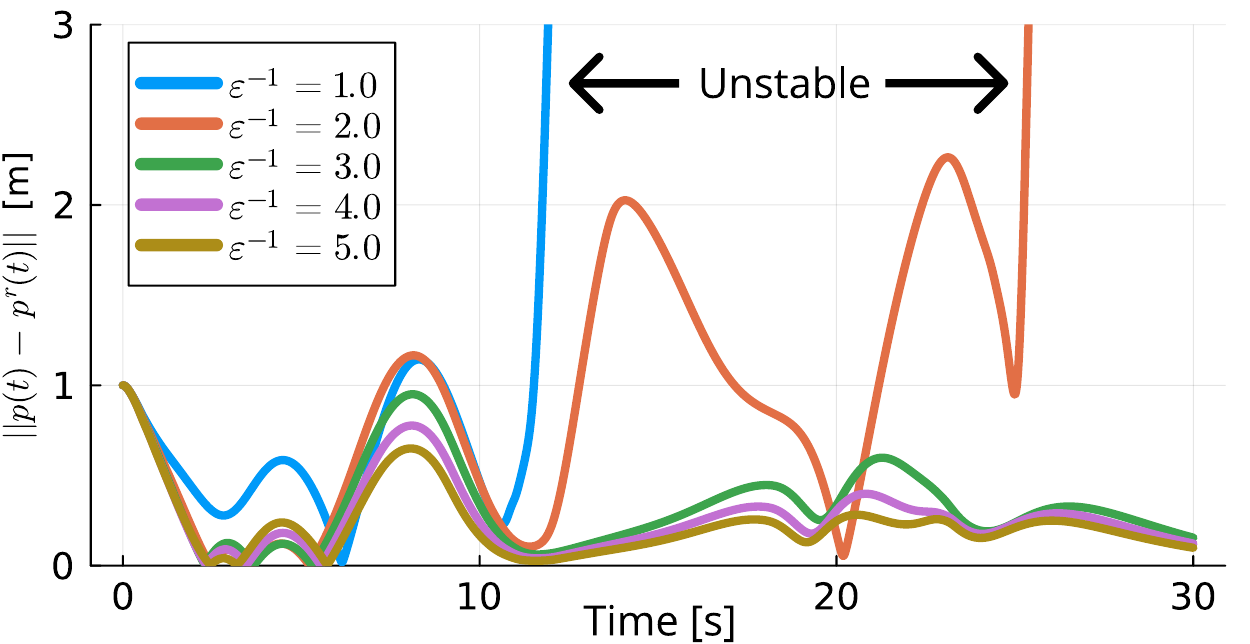}
  \caption{Slow state tracking error vs. time-scale separation $\epsilon$.}
  \label{fig:eps_range}
  \vspace{-7pt}
\end{figure}

To demonstrate the effect of the time-scale separation $\epsilon$, we simulated the landing scenario with five values of $\epsilon$ ranging from 1.0 to 0.125, with an initial vertical position error of 1~m. Maximum error is defined as $\sup_{t\ge2}\norm{p(t)-p^r(t)}$, which is when the initial transient is finished. As expected, smaller $\epsilon$ results in lower tracking error, as seen in \Cref{fig:eps_range}. At larger values of $\epsilon$, the closed-loop system becomes unstable. \Cref{fig:eps_range} and \Cref{tab:eps} confirm the conclusion of \Cref{th:tikh} that the tracking error is reduced as $\epsilon$ shrinks.
\begin{table}
\centering
\caption{Maximum tracking error as $\epsilon \to 0$ }
\begin{tabular}{c||c|c|c|c|c}
\hline
  $\boldsymbol{\epsilon}$ & 1.0 & 0.5 & 0.33 & 0.25 & 0.2 \\ \hline 
  \textbf{max. error [m]} & $\infty$ & $\infty$ & 0.951 & 0.777 & 0.650 \\
  \hline
\end{tabular}
\label{tab:eps}
\vspace{-15pt}
\end{table}

% \begin{remark}
Overall, the \ac{VTOL} is able to cruise, transition, and hover using the proposed adaptive control allocation algorithm, with \emph{no} gain scheduling, blending, or mode switching.
% \end{remark}

% \begin{remark}
% A multilayer perceptron (MLP) neural network can be represented in the form of \eqref{eq:tau} where $W$ represents the linear last layer, $\phi$ is the composition of all the previous layers, and $\delta$ is the approximation or learning error. This representation has been used recently in meta-learning adaptive control~\cite{o2022neural,richards2023control,tang2025meta}.
% \end{remark}

\section{Conclusion}
We have presented an adaptive control framework for uncertain, non-affine in control, underactuated nonlinear systems. Our approach leverages a time-scale separation property together with dynamic control allocation. Using singular perturbation theory and Lyapunov analysis, we established stability guarantees and bounded tracking error of a reference trajectory. Finally, we demonstrated the effectiveness of the proposed method on a \ac{VTOL} quadplane with highly nonlinear dynamics and state-dependent levels of actuation.
% underactuated in level flight and hover, and overactuated during takeoff and landing transition phases.

Future work could involve adaptation to time-varying external disturbances, potentially applying modern machine learning paradigms such as meta learning. Another potential direction is using adaptation algorithms that guarantee exponential stability without persistence of excitation, removing the state predictor and simplifying the architecture. 

{
\setstretch{0.95}
\bibliographystyle{IEEEtran}
\bibliography{biblio, IEEEabrv}
}

\end{document}